\documentclass[conference]{IEEEtran}
\usepackage{amssymb}
\usepackage{amsfonts}
\usepackage{amsmath,amssymb}%
\usepackage[pdftex]{graphicx} 
\usepackage[numbers,sort&compress]{natbib} 
\usepackage{hypernat} 

\newcommand{\E}{\mathbf{E}}

\newcommand{\Prob}{\mathbf{P}}

\newcommand{\Ind}{\mathbf{1}}
\newcommand{\R}{\mathbb{R}}

\newcommand{\dB}{\mathrm{dB}}
\newcommand{\logsd}{\sigma_{\dB}}

\newcommand{\tildeS}{\widetilde{S}}

\newtheorem{thm}{Theorem}
\newtheorem{cor}[thm]{Corollary}
\newtheorem{definition}{Definition}
\newtheorem{lem}[thm]{Lemma}
\newenvironment{lemma}{\bf\begin{lem}\rm\em}{\end{lem}} 
\newtheorem{prop}[thm]{Proposition}
\newtheorem{rem}[thm]{Remark}
\newenvironment{remark}{\bf\begin{rem}\rm}{\end{rem}} 

\usepackage[scriptsize,it]{caption}   
\title{\Huge Equivalence and comparison of heterogeneous cellular networks}\author{{\bf B. B{\l}aszczyszyn}$^*$ and {\bf H.P. Keeler}$^*$ }

\begin{document}

\maketitle

\begin{abstract}
We consider a general heterogeneous network in which, besides general propagation effects (shadowing and/or fading), individual base stations can have different emitting powers and be subject to different parameters of Hata-like path-loss models (path-loss exponent and constant)  due to, for example, varying antenna heights. We assume also that the stations may have varying parameters  of, for example, the link layer performance (SINR threshold, etc). By studying the {\em propagation processes} of signals received by the typical user from all antennas marked by the corresponding antenna parameters, we show that seemingly different heterogeneous networks based on Poisson point processes can be equivalent from the point of view a typical user. These neworks can be replaced with a model where all the previously varying propagation parameters (including path-loss exponents) are set to constants while the only trade-off being the introduction of an isotropic base station density. This allows one to perform analytic comparisons of different network models via their isotropic representations. In the case of a constant path-loss exponent, the isotropic representation simplifies to a homogeneous modification of the constant intensity of the original network, thus generalizing a previous result showing that the propagation processes only depend on one moment of the emitted power and propagation effects. We give examples and applications to motivate these results and highlight an interesting observation regarding random path-loss exponents. 
\end{abstract}


\section{Introduction}
\let\thefootnote\relax\footnotetext{\hspace{-2ex}$^*$Inria/Ens, 23
  av. d'Italie, 75214, Paris, France}
\newcommand{\thefootnote}{\arabic{footnote}}
The rapidly increasing growth of user-traffic in cellular networks is forcing the need for the deployment of multi-tier heterogeneous networks as well as the development of better analytic methods for quantifying their performance. Based on information theoretic arguments, one key performance metric is the signal-to-interference-and-nose-ratio (SINR) experienced by a typical user in the network. The SINR is a function of {\em propagation processes}, which incorporate the distance-dependent path-loss function and (often assumed to be random) fading and/or shadowing, which we refer to as {\em propagation effects}. Consequently, results that cast light on the nature of propagation processes ultimately aid in studying the SINR and other useful characteristics of heterogeneous networks.  

The irregularity of cellular network configurations means that base station positioning is often best assumed to be random, which has motivated the use of models based on stochastic geometry. This assertion has been supported in recent years with tractable models based on the Poisson point process yielding accurate solutions \cite{ANDREWS2011}. Besides the usual tractability and `worst-case' arguments for Poisson models, a recent result~\cite{hextopoi} has shown that a broad range of network configurations give propagation-based results appearing as though the placement of base stations is a Poisson process when sufficiently large log-normal shadowing is incorporated into the model. 

In this work we present two simple yet very useful results on the invariance and equivalence of network characteristics, such as SINR, that are functions of propagation processes. More precisely, we present a marked Poisson model of a random heterogeneous network with the standard power-law path-loss function. Under this model, we assume that the propagation process parameters and base station parameters are all random. We present a generalized  version of {\em propagation (process) invariance}, which shows that propagation processes only depend on  one key moment of the random propagation effects and not their distribution. We list how this result has been used previously in the field of communications.  

Building upon the propagation invariance, we then present the main {\em network equivalence} result stating how a general random heterogeneous network model can be immediately replaced with a `equivalent' network model with the previously random values, including path-loss  exponents, all set to constants. The equivalent `less random' model induces the same propagation process and allows for tractable models of multi-tier cellular networks. This network equivalence allows for the comparison of seemingly quite different heterogeneous models by finding their isotropic and, hence, comparable forms.  Furthermore,  we observe for a constant path-loss exponent  the isotropic representation reduces simply to a homothecy (a modification of the constant intensity) of the original network, which generalizes a previous result showing that the propagation processes only depend on some moment of the emitted power
and propagation effects.  We illustrate the network equivalence result by giving examples and demonstrating how random path-loss exponents effectively change the network density. 

We conclude by discussing possible applications such as deriving $k$-coverage probability expressions for multi-tier networks in the spirit of~\cite{kcovsingle} or deriving more general results that allow the replacement of location-dependent networks with more tractable stationary models. 

\subsection{Related work}
Gilbert and Pollak~\cite{GILBERT1960} derived a classic result showing that a shot noise process, consisting of a sum of functions of Poisson points and some random parameter, remains invariant for many different functions and distributions of random parameters. Lowen and Teich~\cite{LOWEN1990} applied this result to the sum of the power-law functions, akin to the commonly used path-loss function found here, and showed that the sum is independent of the parameter distribution and only relies upon one moment of the random parameter. In the context of  SINR of cellular networks, B\l aszczyszyn  et al.~\cite{blaszczyszyn2010impact} observed this invariance characteristic for interference and propagation losses in general (and not just sums or inteference terms), hence the propagation effects are incorporated into the model by only one moment. Pinto et al.~\cite{PINTO2012} independently derived and used a similar result to show that the node degree of secrecy graphs (based on Poisson processes) is invariant for the distribution of propagation effects. In both papers~\cite{blaszczyszyn2010impact,PINTO2012}, the invariance results are obtained by defining a point process (which we now call propagation process) on the positive real line (a similar process was defined by Haenggi~\cite{HAENGGI2008} but used for different purposes).  More specifically, it was shown that this point process is an inhomogeneous Poisson point process on the positive real line if the underlying base station configuration forms a homogeneous Poisson process. 

In the context of multi-tier (heterogeneous) cellular networks, Dhillon et al.~\cite{DHILLON2012} and Mukherjee~\cite{MUKHERJEE2012} both derived results for the distribution of the (downlink) SINR based on models consisting of independent superpositions of Poisson processes with Rayleigh fading. Madhusudhanan et al.~\cite{MADBROWN2011} obtained similar SINR expressions, but derived and used the above propagation invariance result to show that their (and by extension, the above) results hold for arbitrary propagation effects.  Independently,  B\l aszczyszyn  et al.~\cite{kcovsingle} used the same argument to derive the SINR-based $k$-coverage probability for a single-tier network by first assuming Rayleigh fading, then lifting the assumption via propagation invariance.  It should be stressed that this approach applies to all results based on functions of propagation processes and not just results involving sums or interference terms. A worthy pursuit would be to list all such results that hold under arbitrary propagation effects, but this is beyond the scope of this paper. 

Our second result involves the equivalence of heterogeneous networks with random parameters, including path-loss exponents. For tractability, the aforementioned multi-tier results all assumed constant path-loss exponents across all tiers. Jo et al.~\cite{JO2012} extended this to a model with a different (but constant) path-loss exponent on each tier, but only assumed Rayleigh fading in their work and examined the SINR  based on the base station with the smallest distance to the typical user. Also assuming different (but constant) path-loss exponents, Madhusudhanan et al.~\cite{MADBROWN2012} generalized this approach to arbitrary propagation effects by using propagation invariance. For constant (but different) parameters across all tiers, they also showed that a multi-tier network is stochastically equivalent to a single-tier network with unity parameters while all the original parameters are incorporated into the density of the (inhomogeneous Poisson) propagation process. In the context of cellular networks or related fields, we are unaware of work involving random path-loss exponents or equivalence results to the level of generality (due to more randomized parameters) presented here. 

\section{Model description}
\subsection{Random heterogeneous network}
We outline specifically what we mean by a random heterogeneous\footnote{This term should not be confused with a `nonhomogeneous' or `inhomogeneous' network where the base station density is location-dependent.} network model, which has random base station and path-loss parameters.  On $\R^2$, we model the base stations with a homogeneous Poisson point process
 $\Phi=\{X_i\}$ with density $\lambda$.  We take the `typical user' model approach where one assumes a typical user is located at the origin and consider what he perceives or experiences in the network.  Given $\Phi$, for each $X_i\in \Phi$, let $(P_i,S_i,A_i,\beta_i,T_i)$ be an  independent and identically distributed random vectors with positive (and possibly dependent) coordinates; hence we have an independently marked Poisson process, which we write, with a slight abuse of notation, also as $\Phi$.  For each base station $X_i$, the coordinate $T_i$ represents some parameter\footnote{$T_i$ can be in turn a vector of random parameters for each base station.} dependent on the base station; for example, $T_i$ could be the SINR threshold of the base station, which leads to a generalization of the multi-tier model~\cite{MADBROWN2011,DHILLON2012,MUKHERJEE2012}. For a signal emanating from a base station at $X_i$, let $P_i$ represent the power of the emitted signal whereas $S_i$ represent the propagation effects (shadowing and/or fading) experienced by the typical user.  The random coordinates $A_i$ and $\beta_i$ form part of the (randomly parameterized) path-loss function\footnote{It is often assumed that path-loss exponents $\beta_i>2$ to ensure well-behaved interference in the network.}
\begin{equation}\label{PATHLOSS}
\ell_i(|x|)=A_i|x|^{\beta_i}.
\end{equation}
A practical argument for random parameters $A_i$ and $\beta_i$ stems from path-loss models, such as Hata-based types, that incorporate base station height~\cite[Section 2.7.3]{STUBER2011}. Hence, it is not unreasonable to assume random $A_i$ and $\beta_i$, which lead to more descriptive models. Also the dependence of $A_i$, $\beta_i$, $S_i$, $P_i$ and $T_i$ for a given $X_i$ may be needed to model the situation when the operator tunes $P_i$ and $T_i$ depending on the base station's proposed load.

\subsection{Propagation process}
We define the propagation process, considered as a point process on the positive half-line $\mathbb{R}^{+}$, as
\begin{equation}
\{Y_i\}\equiv \left\{\frac{\ell_i(X_i)}{P_i S_i } \right\}=\left\{\frac{|X_i|^{\beta_i}}{\tildeS_i } \right\},
\end{equation}
where for compactness we often write $\tildeS_i\equiv P_iS_i/A_i$ and we will sometimes omit the subscript; for example, the random vector $(\tildeS,\beta,T)$ is equal in distribution to $(\tildeS_i,\beta_i,T_i)$.  To represent the collection of propagation processes and base station parameters, we introduce the independently marked point process 
\begin{equation}
\Psi\equiv\{(Y_i,T_i)\},
\end{equation}
which we call the {\em marked propagation process}. 
\begin{definition}
We say two heterogeneous network models are (stochastically) equivalent if they induce the same propagation process. 
\end{definition}
In summary, what the typical user `perceives' in our random heterogeneous network is represented by the independently marked point process $\Psi$, which motivates us to seek a method for finding equivalent networks.

\section{Results}
We present a useful lemma, which generalizes a previous result~\cite{blaszczyszyn2010impact,hextopoi} involving a non-random path-loss function. 
\begin{lemma}[Propagation (process) invariance]\label{l.invariance}
Assume that 
\[
\E( \widetilde{S}^{2/\beta}) < \infty.
\]
Then the propagation process $\Psi$ is an independently marked inhomogeneous Poisson point process on $\R^+$ with intensity measure
\begin{align}
\Lambda(s,t) &\equiv\E \left(\#( Y_i,T_i):Y_i\leq s, T_i\leq t\right) \\
&= \lambda \pi \E \left[(s\tildeS_i)^{2/\beta_i} \Ind ( T_i\leq t ) \right], 
\end{align}
where  $\#$ denotes the counting measure.
\end{lemma}

\begin{proof} 
By the displacement theorem~\cite[Section 1.3.3]{FnT1} and Campbell's theorem~\cite[Corolloary 2.2]{FnT1}, $\Psi$ is a Poisson point process with intensity measure
\begin{align}
\Lambda(s,t)&=\E[\sum_{(Y_i,T_i)\in\Phi} \Ind ( Y_i\leq s, T_i\leq t )] \\
&=\lambda \E \int_{\R^2}\Ind (Y_i\leq s) \Ind (  T_i\leq t ) dx \\
&=\lambda (2\pi)\E \int_0^{\infty}\Ind (r \leq (s\tildeS_i)^{1/\beta_i}) \Ind ( T_i\leq t ) rdr,
\end{align}
where integrating completes the proof.

\end{proof} 
\begin{remark} 
Consider the intensity function 
\[
\Lambda(s)=\lim_{t\rightarrow\infty}{\Lambda(s,t)},
\]
For an intuitive meaning of the $\Lambda$, the quantity $\Lambda(y_2)-\Lambda(y_1)$
can be interpreted as the mean number of base stations in the network received by the typical user with signal power between $1/y_2$ and $1/y_1$.
\end{remark} 
\begin{remark}
It is often assumed the propagation effects of $S$ are Rayleigh in
connection to the convenient properties of the resulting exponential
distributions of $S$ and their connection to Laplace
transforms. However, in the case of $\beta$ being equal to some
constant, then propagation invariance implies that $\Lambda$ depends
only on $\E( \widetilde{S}^{2/\beta}) $ but not on the type of
distribution of $S$. This holds for general functions of the
propagation process (example $\max\{Y_i\}$), and not just sums or
inteference terms. Hence, propagation effects can be represented by
setting $\tildeS$ to a constant and replacing $\lambda$ with
$\lambda'=\lambda\E( \tildeS^{2/\beta})$. One can also
(e.g. for pure mathematical convenience) assume exponential (say mean one) 
propagation effects and replace $\lambda$ with $\lambda'=\lambda\E(
\tildeS^{2/\beta})/\Gamma(2/\beta+1)$, where $\Gamma(2/\beta+1)$ is
the $2/\beta\,$th moment of exponential, variable of mean one. 
 For convenience we have included a table of  closed-form expressions
 of this moment for commonly used  fading and shadowing distributions (Table \ref{ESbeta}).
\end{remark}
%

\begin{table}
    \begin{tabular}{| p{1.6cm} | p{2.8cm}  | p{3cm}  |}
    \hline
    Distribution &Probablity density of $S$ & $\E(S^{2/\beta})$ \\ \hline
    Log-normal  &\shortstack{$1/(s\sqrt{2\pi}\sigma)$  \\$\times e^{-(\ln s-\mu)^2/(2\sigma^2) }	$} & $e^{2(\sigma^2+\mu\beta)/\beta^2}$ \\ \hline
    Exponential&$ \lambda e^{-\lambda s}$ & $\lambda^{-2/\beta}\Gamma(2/\beta+1)$  \\ \hline
    Weibull &\shortstack{ $(k/\lambda)(s/\lambda)^{k-1} $\\ $\times e^{-(s/\lambda)^{k} }$} & $\lambda^{2/\beta}\Gamma(2/(\beta k)+1) $ \\   \hline
   Nakagami &\shortstack{ $ \frac{ 2m^{m}}{\Gamma(m)\Omega^m}  s^{2m-1} $\\ $\times e^{-(m/\Omega)s^{2} }$} & $\frac{\Gamma(1/\beta+m)}{\Gamma(m)}\left(\frac{\Omega}{m}\right)^{1/\beta} $ \\   \hline
    Rice & \shortstack{$  (s/\sigma^2) I_0(s\nu / \sigma^2)  $ \\ $ \times  e^{-(s^2+\nu^2)/(2\sigma^2)}  $ }
& \shortstack{$ (2 \sigma^2)^{1/\beta} \Gamma(1/\beta+1) $
\\$ \times {}_1F_1\left(-1/\beta,1;-\frac{\nu^2}{2\sigma^2}\right)  $ }  \\ \hline
    \end{tabular}
\caption{ $\Gamma$ is the gamma function and $I_0$ is the modiffied Bessel function of the first kind. }\label{ESbeta}
\end{table}


Before presenting the main result of this paper, which gives a natural way of finding equivalent networks, we introduce the following functions
\begin{equation}
F_r'(t)=\frac{\phi(r,t)}{\phi(r)},
\end{equation}
where
\begin{equation}\label{e.phirt}
\phi(r,t)=\frac{1}{2\pi r }\frac{\partial }{\partial r}\Lambda(r^{\beta'},t) ,
\end{equation}
and 
\begin{equation}\label{e.phir}
\phi(r)=\lim_{t\rightarrow \infty }\phi(r,t),
\end{equation}
for some $\beta'>0$, and we implicitly assume that $0< \phi(r) <\infty$ and $F_r(t)$ is non-decreasing in $t$. \\
We now present the main result, which shows a stochastic equivalence between different heterogeneous networks in terms of propagation processes experienced by the typical user.  

\begin{prop}[Network equivalence]
Assume for some constant $\beta'>0$ and all finite $s>0$ that
\begin{equation}
\E \left[(s\tildeS_i)^{2\beta'/\beta}  \right] < \infty.
\end{equation}
Then $\Psi$ is stochastically equivalent to a location-dependent, independently marked propagation process $\Psi'=\{(Y_i',T_i')\}_i$ induced by an isotropic (possibly inhomogeneous) Poisson point process $\Phi'=\{X_i'\}$ on $\R^2$ with spatial intensity $\phi(|x|)dx$ of base stations, for which one sets the constant values $P_i\equiv S_i\equiv A_i\equiv1$ and $\beta_i\equiv\beta'$, and independently random marks $T_i'$ whose distribution, when associated to point $X_i$, is given by $F_r'(t)$ when $|X_i'|=r$.
\end{prop}

\begin{proof} 
If we consider $\Phi'$ and set $P_i =S_i= A_i=1$, $\beta_i=\beta'$ and $T_i=T_i'$,   then the corresponding propagation process is a Poisson process with intensity
\begin{align}
\Lambda'(s,t)&=2\pi \E \int_0^{s^{1/\beta'}} \Ind ( T_i'(r)\leq t ) \phi(r) rdr\\
&=2\pi \int_0^{s^{1/\beta'}} F_r' (t ) \phi(r) rdr.
\end{align}
Differentiating the above integral with respect to $r$ recovers (\ref{e.phirt}), which in the limit $t\rightarrow\infty$ gives $\phi(r)$.


\end{proof} 
\begin{remark}
For some original heterogeneous network $\Phi$ with a propagation process $\Psi$,  the above results says that $\Phi$ will have an {\em isotropic representation} $\Phi'$ that (stochastically) has the same typical user propagation process $\Psi$. Furthermore, if the original network $\Phi$ has all $\beta_i$ equal to some \em{constant}, then $\Phi'$ will have a constant density denoted by $\lambda'$, hence $\Phi'$ is also homogeneous.
\end{remark}

\section{Examples}
We now illustrate our main result by covering three examples of network models, for each of which we find the isotropic representation. 

\subsection{Free-space compensation}
Consider an original network $\Phi$, with all $\beta_i=\bar{\beta}$ for some constant $\bar{\beta}>2$ and all $\tildeS_i=1$. For the isotropic representation $\Phi'$, set the relative path-loss exponent $\beta'=2$ (the `free-space' value), which implies the density
\begin{equation}\label{e.freespacephi}
\phi(|x|)=|x|^{2(2/\bar{\beta}-1)} .
\end{equation}
In other words, if one (virtually) assumes a free-space path-loss model when the (original) path-loss exponent is $\bar{\beta}$, then in order to obtain the equivalent propagation process one needs to compensate the equivalent model by assuming the isotropic power-law network density (\ref{e.freespacephi}) . Note that the density of the (free-space) isotropic representation $\Phi'$ decreases as $|x|$ increases due to $\bar{\beta}>2 \implies 2(2/\bar{\beta}-1)<0$. Similarly, one can compare networks with other $\beta$ values instead of the free-space value (for example, values corresponding to urban and suburban environments). 

\subsection{Propagation effects imply sparser networks}\label{sparser}
We consider two original networks $\Phi_1$ and $\Phi_2$ with the identical density $\lambda$ and path-loss parameters, which are all set to constants $\beta_i=\bar{\beta}>2$ and $A_i=\bar{A}>0$. Given $\Phi_1$, we assume each $X_i\in\Phi_1$ has an identically distributed {\em random} propagation variable $\tildeS_1$. Conversely, given $\Phi_2$, we assume that each $X_i\in\Phi_2$ has a {\em constant} propagation variable $\tildeS_2=\E(\tildeS_i)$. Then for $\beta'=\bar{\beta}$,  the two corresponding isotropic representations $\Phi_1'$ and $\Phi_2'$ have the respective densities 
\[
\phi_1(|x|)\equiv\lambda_1'=\lambda \E( \tildeS^{2/\beta}), \quad \phi_2(|x|)\equiv\lambda_2'=\lambda \E( \tildeS)^{2/\beta}.
\] 
Jensen's inequality implies $\E( \tildeS^{2/\beta})\leq \E( \tildeS)^{2/\beta}$ , hence a network with random propagation effects is equivalent to a {\em sparser} network without propagation effects.  More generally, for two propagation variables $ \tildeS_1$ and $ \tildeS_2$ with equal mean, if $ \tildeS_2$ is more variable than $ \tildeS_1$, formally defined by the  stochastic ordering $ \tildeS_1  \leq _{cx} \tildeS_2 $ (that is, $E[ f(\tildeS_1)] \leq f[E( \tildeS_2)]$ for all convex $f$), then $\Phi_2'$ is sparser than $\Phi_1'$. In other words, more variability in $ \tildeS$  effectively implies a sparser network.

\subsection{Two-tier network}
We now consider a two-tier Poisson network model where the path-loss parameters $A$ and $\beta$ are set to constants, which is a case of the multi-tier model~\cite{MADBROWN2011,MUKHERJEE2012,DHILLON2012}, but with the $\beta$ values depending on each tier. Specifically, the first tier is a Poisson process $\Phi_1$ with density $\lambda_1$, and given $\Phi_1$, each $X_i$ has the vectors of constant values $(\beta_1,A_1)$ and random $\tildeS_1$ and $T_1$. Similarly, the second tier is $\Phi_2$ with density $\lambda_2$ and parameters $(\beta_2,A_2)$ and random $\tildeS_2$ and $T_2$.  The resulting propagation process is a Poisson process with intensity
\begin{align}
\Lambda(s,t)&=\pi \lambda_1\E(\tildeS_1^{2/\beta_1}) \Ind (t_1\leq t) s^{2/\beta_1}\\
&+\pi
\lambda_2 \E( \tildeS_2^{2/\beta_2}) \Ind (t_2\leq t) s^{2/\beta_2}   .
\end{align}
Then this two-tier network is equivalent to a single-tier network $\Phi'$ ($\tildeS=1$) with spatial density $\phi(|x|)$, which for some arbitrary $\beta'$ is given by
\begin{align}
\phi(r)&=\frac{\beta'}{\beta_1} \lambda_1 \E( \tildeS_1^{2/\beta_1}) r^{2(\beta'/\beta_1-1)}\\
&+\frac{\beta'}{\beta_2} \lambda_2  \E( \tildeS_2^{2/\beta_2})  r^{2(\beta'/\beta_2-1)}  .
\end{align}
The independent marks $T_r'$ have a mixed distribution dependent on the distance from point $x$ to the origin, $|x|=r$, namely
\begin{equation}
\Prob( T_r' \leq t ) = p_1(r)\Prob(T_1 \leq t ) + p_2(r)\Prob(T_2 \leq t )  
\end{equation}
where $p_1(r)$ and $p_2(r)$ are the probabilities that a $T$ value  belongs to the first or second tier respetively 
\begin{align}
		p_1(r)=  &  \frac{\beta'}{\beta_1\phi(r)} \lambda_1\E( \tildeS_1^{2/\beta_1}) r^{2(\beta'/\beta_1-1)} \\
		p_2(r)= &   \frac{\beta'}{\beta_2\phi(r)} \lambda_2\E( \tildeS_2^{2/\beta_2})  r^{2(\beta'/\beta_2-1)}.
\end{align}

These results say that, in terms of propagation processes, this two-tier network behaves as a single isotropic network with random location-dependent marks $T_r'$. A natural choice for $\beta'$ is $\beta'=(\beta_1+\beta_2)/2$. In this case, we note that if $\beta_1=\beta_2$ (as in~\cite{DHILLON2012}), then $\Phi'$ is homogeneous.

\subsubsection{Numerical results}
We demonstrate this example further with some numerical values. For the $\beta$ and $A$ values of the two networks, we use the COST231-Hata model outlined in \cite[Section 2.7.3.3]{STUBER2011} with the parameters for a metropolitan area in a large city. For both tiers in the two-tier model, the user  height is $1$m and the carrier frequency is $1800$ MHz. However, in the tier-one and tier-two we specify that the antennas have different heights, $h_1=20$m and $h_2=100$ m respectively, which correspond to model parameters $\beta_1=3.638$, $A_1=1.986\times10^{14}$, $\beta_2=3.180$ and $A_2=2.148\times10^{13}$.  The densities of the two networks are $\lambda_1=1.8$ and $\lambda_2=2.2$.

For comparison purposes, we now consider a single-tier network $\Phi_3$ with density $\lambda_3 =\lambda_1+\lambda_2$, which is, in a way, the average of the two-tier network. We use the same COST231-Hata parameters but we set the antenna height to the (spatial) average of the two antenna heights, $h_3=(\lambda_1h_1+\lambda_2h_2)/(\lambda_3)= 64$m, giving $\beta_3= 3.307$ and $A_3=3.979\times10^{13}$. We compare the two-tier and one-tier network by using the equivalent isotropic representation with $\beta'=\beta_3$.

When plotting we multiply $\phi(r)$ by $ E(A^{\beta/2})=A^{\beta/2}$
so that the resulting quantities have the same magnitude as the
original base station densities. For both models, we plot
$\phi(r)\E(A^{\beta/2})$ with and without log-normal shadowing $S$ of
mean one and $5$ dB logarithmic standard
deviation~\footnote{\label{ftnte} $S=\exp[-\sigma^2/2+\sigma N]$,
  where $N$ is standard Gaussian variable;  the {\em logarithmic
    standard deviation} of $S$ is given by $\logsd=\sigma10/\log10$.}.

The numerical results (Fig. \ref{f.PhiRescaled}) show that random shadowing $S$ effectively corresponds to a sparser network while random path-loss exponent $\beta$ effectively increases the base station density around the typical user. One can conjecture that a stochastic ordering exists for $\beta$ akin to that for the propagation variable $\tildeS$ (in Section \ref{sparser}), but that is beyond the scope of this paper.

 
\begin{figure}[t!]
\begin{center}
\centerline{\includegraphics[width=1\linewidth]{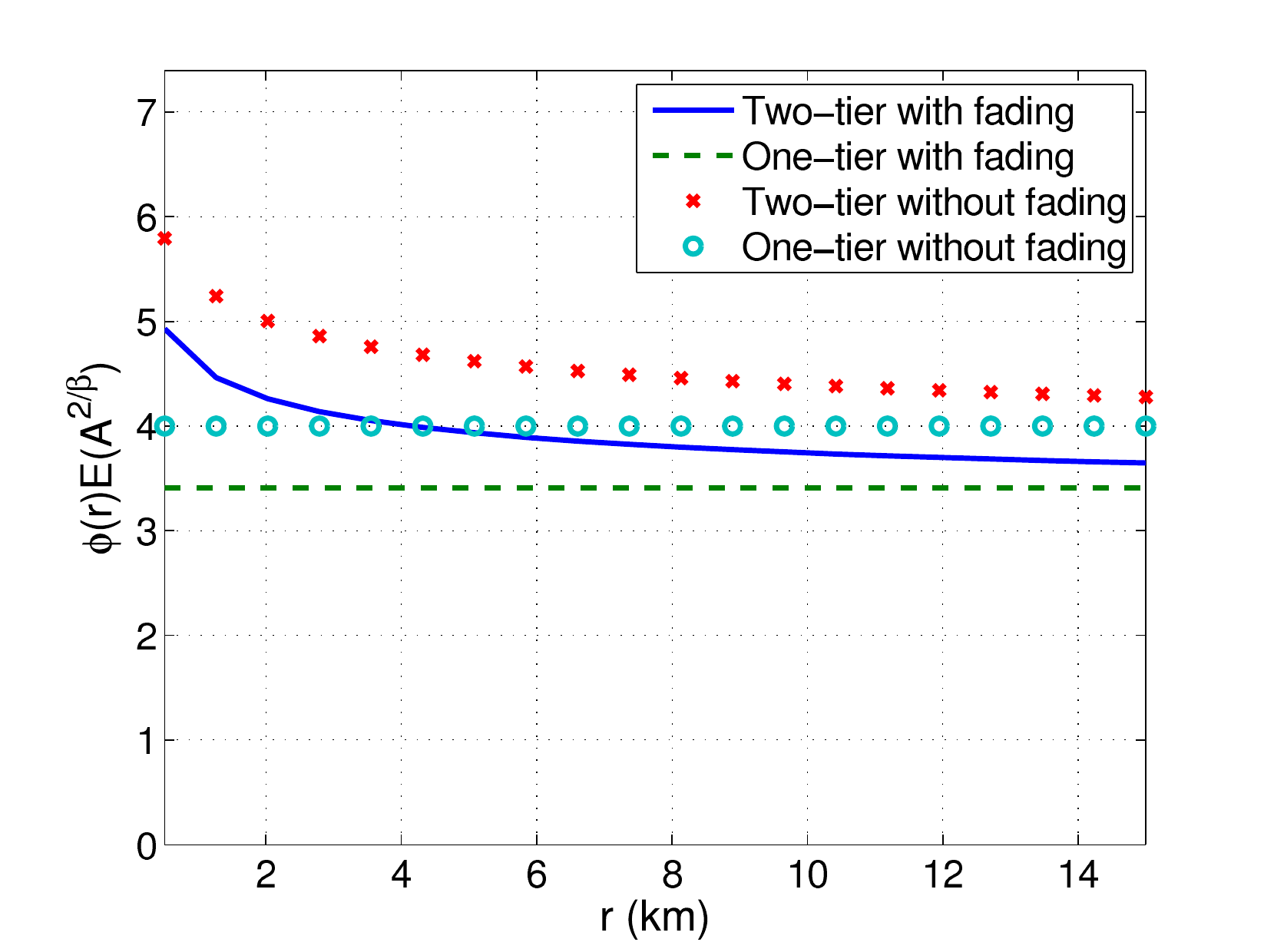}}
\caption{Two networks under the COST231-Hata path-loss model with and without log-normal shadowing of $5$ dB logarithmic standard deviation. For the typical user, random propagation effects and path-loss exponent effectively make the network sparser or denser  respectively. 
\label{f.PhiRescaled}}
\end{center}
\vspace{-2ex}
\end{figure}

\section{Future directions and conclusion}
We discussed the importance of obtaining tractable models for SINR, which is a function of propagation processes. For each base station $X_i$, we can assume that $T_i$ is a random variable that represents the SINR threshold, which is the technology-dependent level that the SINR must exceed to establish a connection. The equivalence network proposition then allows for a multi-tier network model to be represented as a single-tier model, which can lead to deriving SINR expressions for heterogeneous networks where each tier (or even base station) has a different SINR threshold, thus extending previous $1$-coverage probability results~\cite{DHILLON2012}, in a manner akin to~\cite{kcovsingle}, to the $k$-coverage case for multi-tier networks. 

It is sometimes argued that Poisson processes do not model certain (or all) tiers of heterogeneous networks adequately for they fail to capture clustering or repulsion of base stations. However, if sufficiently large log-normal shadowing is present, then it has been recently shown that a wide class of network configurations with constant densities can be approximated with Poisson processes~\cite{hextopoi}. Hence, Poison processes can still model all the tiers even if some clustering or repulsion exists, thus allowing one to apply the equivalence result to find equivalent single-tier networks. 

This paper has covered re-interpreting single-tier or multi-tier
network models (with each tier consisting of a homogeneous Poisson
process) into their equivalent  isotropic forms. Conversely, it is not unreasonable to assume that cellular networks in cities may have isotropic base station densities, and then one wanting to know their equivalent homogeneous forms. This reversed setting motivates the need for a more generalized version  of our network equivalence proposition, thus transforming isotropic networks to those with constant base station densities. 

In summary, for the power-law path-loss model with arbitrary propagation effects, we have shown that a wide class of heterogeneous network models based on Poisson processes are equivalent in terms of propagation processes perceived by the typical user.

\bibliographystyle{IEEEtran}

\begin{thebibliography}{10}
\providecommand{\url}[1]{#1}
\csname url@rmstyle\endcsname
\providecommand{\newblock}{\relax}
\providecommand{\bibinfo}[2]{#2}
\providecommand\BIBentrySTDinterwordspacing{\spaceskip=0pt\relax}
\providecommand\BIBentryALTinterwordstretchfactor{4}
\providecommand\BIBentryALTinterwordspacing{\spaceskip=\fontdimen2\font plus
\BIBentryALTinterwordstretchfactor\fontdimen3\font minus
  \fontdimen4\font\relax}
\providecommand\BIBforeignlanguage[2]{{%
\expandafter\ifx\csname l@#1\endcsname\relax
\typeout{** WARNING: IEEEtran.bst: No hyphenation pattern has been}%
\typeout{** loaded for the language `#1'. Using the pattern for}%
\typeout{** the default language instead.}%
\else
\language=\csname l@#1\endcsname
\fi
#2}}

\bibitem{ANDREWS2011}
J.~Andrews, F.~Baccelli, and R.~Ganti, ``A tractable approach to coverage and
  rate in cellular networks,'' \emph{IEEE Trans. Commun.}, vol.~59, no.~11, pp.
  3122 --3134, november 2011.

\bibitem{hextopoi}
B.~B{\l}aszczyszyn, M.~Karray, and H.~Keeler, ``Using {P}oisson processes to
  model lattice cellular networks,'' in \emph{Proc. of IEEE INFOCOM}, 2013,
  available also at~\url{http://arxiv.org/abs/1207.7208}.

\bibitem{kcovsingle}
------, ``{SINR}-based k-coverage probability in cellular networks with
  arbitrary shadowing,'' in \emph{ISIT 2013 IEEE International Symposium on
  Information Theory}, 2013, available also
  at~\url{http://arxiv.org/abs/1301.6491}.

\bibitem{GILBERT1960}
E.~N. Gilbert and H.~O. Pollak, ``{Amplitude distribution of shot noise},''
  \emph{Bell Systems Technical Journal}, vol.~39, pp. 333--350, 1960.

\bibitem{LOWEN1990}
S.~Lowen and M.~C. Teich, ``Power-law shot noise,'' \emph{Information Theory,
  IEEE Transactions on}, vol.~36, no.~6, pp. 1302--1318, 1990.

\bibitem{blaszczyszyn2010impact}
B.~Blaszczyszyn, M.~Karray, F.~Klepper, \emph{et~al.}, ``Impact of the
  geometry, path-loss exponent and random shadowing on the mean interference
  factor in wireless cellular networks,'' in \emph{Third Joint IFIP Wireless
  and Mobile Networking Conference (WMNC)}, 2010.

\bibitem{PINTO2012}
P.~Pinto, J.~Barros, and M.~Win, ``Secure communication in stochastic wireless
  networks -- {P}art {I}: Connectivity,'' \emph{Information Forensics and
  Security, IEEE Transactions on}, vol.~7, no.~1, pp. 125--138, 2012.

\bibitem{HAENGGI2008}
M.~Haenggi, ``A geometric interpretation of fading in wireless networks: Theory
  and applications,'' \emph{Information Theory, IEEE Transactions on}, vol.~54,
  no.~12, pp. 5500--5510, 2008.

\bibitem{DHILLON2012}
H.~Dhillon, R.~Ganti, F.~Baccelli, and J.~Andrews, ``Modeling and analysis of
  {K}-tier downlink heterogeneous cellular networks,'' \emph{IEEE J. Sel. Areas
  Commun.}, vol.~30, no.~3, pp. 550--560, april 2012.

\bibitem{MUKHERJEE2012}
S.~Mukherjee, ``Distribution of downlink {SINR} in heterogeneous cellular
  networks,'' \emph{Selected Areas in Communications, IEEE Journal on},
  vol.~30, no.~3, pp. 575--585, 2012.

\bibitem{MADBROWN2011}
P.~Madhusudhanan, J.~Restrepo, Y.~Liu, T.~Brown, and K.~Baker, ``Multi-tier
  network performance analysis using a shotgun cellular system,'' in
  \emph{Global Telecommunications Conference (GLOBECOM 2011), 2011 IEEE}, 2011,
  pp. 1--6.

\bibitem{JO2012}
H.-S. Jo, Y.~J. Sang, P.~Xia, and J.~Andrews, ``Heterogeneous cellular networks
  with flexible cell association: A comprehensive downlink sinr analysis,''
  \emph{Wireless Communications, IEEE Transactions on}, vol.~11, no.~10, pp.
  3484--3495, 2012.

\bibitem{MADBROWN2012}
P.~Madhusudhanan, J.~Restrepo, Y.~Liu, and T.~Brown, ``Downlink coverage
  analysis in a heterogeneous cellular network,'' in \emph{Global
  Communications Conference (GLOBECOM), 2012 IEEE}, 2012, pp. 4170--4175.

\bibitem{STUBER2011}
G.~L. St$\ddot{\mathrm{u}}$ber, \emph{Principles of Mobile Communication},
  2nd~ed.\hskip 1em plus 0.5em minus 0.4em\relax NYC, NY, USA: Springer, 2011.

\bibitem{FnT1}
F.~Baccelli and B.~B{\l}aszczyszyn, \emph{Stochastic Geometry and Wireless
  Networks, Volume I --- Theory}, ser. Foundations and Trends in
  Networking.\hskip 1em plus 0.5em minus 0.4em\relax NoW Publishers, 2009, vol.
  3, No 3--4.

\end{thebibliography}

\end{document}